\documentclass[aps,prl,floatfix,twocolumn,tightenlines,amsmath,amssymb,nofootinbib]{revtex4-1}

\usepackage{graphicx}
\usepackage{amssymb}
\usepackage{color}
\usepackage{psfrag}
\usepackage{ifsym}
\usepackage{epstopdf}
\usepackage{dsfont}
\usepackage{amsmath}
\usepackage{multirow} 
\usepackage{paralist}

\usepackage{amsthm}
\usepackage{enumitem}

\newcommand{\ket}[1] {| #1 \rangle}
\newcommand{\braket}[2] {\langle #1 | #2 \rangle}
\newcommand{\ketbra}[1]{ | #1 \rangle\!\langle #1 |}

\newcommand{\one}{\leavevmode\hbox{\small1\normalsize\kern-.33em1}}

\newcommand{\ba}{\begin{eqnarray}}
\newcommand{\ea}{\end{eqnarray}}
\newtheorem{thm}{Theorem}

\newtheorem{lem}[thm]{Lemma}
\newtheorem*{lem*}{Lemma}

\begin{document}

\title{How $\psi$-epistemic models fail at explaining the indistinguishability of quantum states}

\author{Cyril Branciard}
\affiliation{School of Mathematics and Physics, The University of Queensland, Brisbane, QLD 4072, Australia\\
Institut N\'eel, CNRS and Universit\'e Grenoble Alpes, 38042 Grenoble Cedex 9, France}

\date{\today}

\begin{abstract}
We study the extent to which $\psi$-epistemic models for quantum measurement statistics---models where the quantum state does not have a real, ontic status---can explain the indistinguishability of nonorthogonal quantum states. This is done by comparing the overlap of any two quantum states with the overlap of the corresponding classical probability distributions over ontic states in a $\psi$-epistemic model.
It is shown that in Hilbert spaces of dimension $d \geq 4$, the ratio between the classical and quantum overlaps in any $\psi$-epistemic model must be arbitrarily small for certain nonorthogonal states, suggesting that such models are arbitrarily bad at explaining the indistinguishability of quantum states.
For dimensions $d = 3$ and 4, we construct explicit states and measurements that can be used experimentally to put stringent bounds on the ratio of classical-to-quantum overlaps in $\psi$-epistemic models, allowing one in particular to rule out maximally $\psi$-epistemic models more efficiently than previously proposed.
\end{abstract}

\maketitle

Despite its central role, the quantum state remains one of the most mysterious objects of quantum theory.
Does it correspond to any physical reality, or does it merely represent one's information on a quantum system? These questions have triggered intense debates among physicists and philosophers since the advent of quantum theory, and are still the subject of active research in the study of quantum foundations.

The general framework of \emph{ontological models}~\cite{harrigan2010ein} proposes a rigorous approach to address such questions.
This framework presupposes the existence of underlying states of physical reality---\emph{ontic states}---and describes the quantum state as a state of knowledge---an \emph{epistemic state}---about the actual ontic state of a given quantum system, represented by a probability distribution over the set of ontic states.
A fundamental distinction is made between so-called \emph{$\psi$-ontic models}, in which any underlying ontic state determines the quantum state uniquely, and so-called \emph{$\psi$-epistemic models}, where the same ontic state can be compatible with different quantum states; i.e., the probability distributions corresponding to two different quantum states can overlap.
In the former case, the ontic state ``encodes'' the quantum state, which can hence be understood as a physical property of the system, while in the latter case the quantum state cannot be given the status of a real physical property.

The epistemic view of the quantum state is quite attractive, as it gives a natural explanation to many puzzling quantum phenomena~\cite{spekkens2007evi}---including, for instance, the collapse of the wave function, or the impossibility to perfectly distinguish nonorthogonal quantum states. However, it has been shown that $\psi$-epistemic models must be severely constrained if they are to reproduce the statistics of quantum measurements. Notably, Pusey, Barrett, and Rudolph showed that no $\psi$-epistemic models satisfying some natural independence condition for composite systems can reproduce quantum predictions~\cite{pusey2012ont}. A number of other no-go theorems have since been proven, under various assumptions~\cite{hall2011gen,colbeck2012is-,miller2013alt,schlosshauer2012imp,hardy2013are, maroney2012how, maroney2012abn,leifer2013max,nigg2012can, patra2013no-, aaronson2013psi,schlosshauer2014no-,barrett2013npe,colbeck2013a-s,leifer2014pem}.
Explicit $\psi$-epistemic models have nevertheless been constructed for quantum measurement statistics in Hilbert spaces of any dimension~\cite{lewis2012dis,aaronson2013psi}, which get around the assumptions of these no-go theorems.

While it is thus impossible to completely rule out $\psi$-epistemic models for quantum theory without auxiliary assumptions, one can still set bounds on the extent to which $\psi$-epistemic models can possibly explain certain quantum features, and, in particular, the indistinguishability of nonorthogonal quantum states~\cite{maroney2012how,maroney2012abn,leifer2013max,barrett2013npe,leifer2014pem}.
This can be done by investigating how much of  the overlap of two quantum states can be accounted for by the classical overlap of their corresponding epistemic states.
In this respect, Maroney and Leifer~\cite{maroney2012how,maroney2012abn,leifer2013max} showed that a certain (asymmetric) measure of overlap of classical probability distributions had to be smaller than $|\braket{\phi}{\psi}|^2$ for some quantum states $\ket{\psi}, \ket{\phi}$ in any Hilbert space of dimension $d \geq 3$. Barrett \emph{et al.}~\cite{barrett2013npe} showed that the ratio between two directly comparable measures for the classical and quantum overlaps had to scale at most like $1/d$ for certain pairs of quantum states when the dimension $d$ increases, while Leifer~\cite{leifer2014pem} exhibited states for which the same ratio has to decrease exponentially with $d$. Following these works, and in particular the approach of Barrett \emph{et al.}~\cite{barrett2013npe}, we show in this Letter that for all dimensions $d \geq 4$, there actually exist states for which the ratio of classical-to-quantum overlaps must be arbitrarily small---meaning that $\psi$-epistemic models are arbitrarily bad at explaining the indistinguishability of certain nonorthogonal quantum states in dimensions $d \geq 4$.
For dimensions 3 and 4, we exhibit explicit states and measurements that would allow one to bound experimentally the ratio of classical-to-quantum overlaps, in a more efficient way than previously proposed~\cite{barrett2013npe}.

\paragraph{Ontological models.---}

Let us start by recalling the framework of ontological models~\cite{harrigan2010ein,footnote_mixed_POVM}.
Such models posit the existence of \emph{ontic states} $\lambda$, taking values in some measurable space $\Lambda$, meant to describe the real state of affairs of a given physical system~\cite{footnote_HV}. The preparation of a quantum state $\ket{\psi}$ corresponds to the preparation of an ontic state $\lambda$ according to some probability distribution---called \emph{epistemic state}---$\mu_\psi(\lambda)$, such that
\ba
\mu_\psi(\lambda) \geq 0 \quad \text{and} \quad \int_\Lambda \mu_\psi(\lambda) \, \text{d}\lambda = 1.
\ea
In the following, the notation $\psi$ will be used to refer to both the quantum state $\ket{\psi}$ and the corresponding epistemic state $\mu_\psi$.

When a measurement $M$ is performed, the probability for each of its possible outcomes $m$ is supposed to depend only on the ontic state, and is given by a response function $\xi_M(m|\lambda)$. For the response function to be a well-defined probability distribution for each ontic state $\lambda$, it must be nonnegative and normalized; i.e., it satisfies
\ba
\xi_M(m|\lambda) \geq 0 \quad \text{and} \quad \sum_m \, \xi_M(m|\lambda) = 1. \label{eq:norm_sum_xi_m}
\ea
When the measurement is performed on a system prepared in the epistemic state $\mu_\psi$, the probability of each outcome is given, after integrating over $\Lambda$, by
\ba
P_M(m|\psi) = \int_\Lambda \xi_M(m|\lambda) \, \mu_\psi(\lambda) \, \text{d}\lambda.
\ea
The model reproduces the quantum measurement statistics in a $d$-dimensional Hilbert space ${\cal H}$ if for any state $\ket{\psi} \in {\cal H}$ and any projection eigenbasis $M = \{ \ket{m} \}$ of ${\cal H}$, the above probabilities $P_M(m|\psi)$ are those predicted by quantum theory---namely, according to the Born rule, $P_M(m|\psi) = |\braket{m}{\psi}|^2$.

In this framework, an ontological model is said to be \emph{$\psi$-epistemic} if there exists at least one pair of different (pure) quantum states $\ket{\psi}, \ket{\phi}$, for which the corresponding probability distributions $\mu_\psi, \mu_\phi$ have nonzero overlap; otherwise, the model is said to be \emph{$\psi$-ontic}.
Hence, in a $\psi$-epistemic models a single ontic state $\lambda$ could be compatible with more than one epistemic state. This suggests that the impossibility of perfectly distinguishing two nonorthogonal quantum states may be (at least partially) explained by the fact that the underlying real state of affairs may sometimes be the same for the two states.

To quantify this, we shall compare, as in Refs.~\cite{nigg2012can,barrett2013npe,leifer2014pem}, the probability of successfully distinguishing the two quantum states $\ket{\psi}, \ket{\phi}$ using optimal quantum measurements, and that of distinguishing the two epistemic states $\mu_\psi, \mu_\phi$ given that one knows the ontic state $\lambda$ (assuming in each case that $\ket{\psi}$ and $\ket{\phi}$, respectively $\mu_\psi$ and $\mu_\phi$, have been prepared with equal probabilities). These probabilities of success are given, respectively, by $1 - \omega_Q(\ket{\psi}, \ket{\phi})/2$ and $1 - \omega_C(\mu_\psi, \mu_\phi)/2$, where the \emph{quantum} and \emph{classical overlaps} $\omega_Q$ and $\omega_C$ are defined as~\cite{barrett2013npe}
\ba
\omega_Q(\ket{\psi}, \ket{\phi}) &=& 1 - \sqrt{1 - |\braket{\psi}{\phi}|^2}, \\[1mm]
\omega_C(\mu_\psi, \mu_\phi) &=& \int_\Lambda \min\! \big[ \mu_\psi(\lambda) , \mu_\phi(\lambda) \big] \, \text{d}\lambda.
\ea
Clearly, one has $0 \leq \omega_C(\mu_\psi, \mu_\phi) \leq \omega_Q(\ket{\psi}, \ket{\phi}) \leq 1$ in any model that reproduces quantum measurement statistics~\cite{nigg2012can}. Indeed, given $\lambda$ one can reproduce the optimal quantum measurement that gives a probability $1 - \omega_Q(\ket{\psi}, \ket{\phi})/2$ of distinguishing the two preparations; the optimal classical strategy for distinguishing $\mu_\psi$ and $\mu_\phi$ cannot give a lower probability of success. Ontological models such that $\omega_C(\mu_\psi, \mu_\phi) = \omega_Q(\ket{\psi}, \ket{\phi})$ for all states $\psi, \phi$ are said to be \emph{maximally $\psi$-epistemic}~\cite{footnote_max_psi_epist}: these models fully explain the indistinguishability of nonorthogonal quantum states by that of the corresponding epistemic states.
In general, for any two nonorthogonal states $\psi, \phi$ (such that $\omega_Q(\ket{\psi}, \ket{\phi}) \neq 0$) we shall define the \emph{ratio of classical-to-quantum overlaps} as
\ba
\kappa(\psi, \phi) = \frac{\omega_C(\mu_\psi,\mu_\phi)}{\omega_Q(\ket{\psi}, \ket{\phi})} \, \leq \, 1. \label{eq:def_kappa}
\ea
As the quantum and classical overlaps have the same operational interpretation in terms of probabilities of successful distinctions, this ratio suitably quantifies how much the model explains of the quantum indistinguishability of $\ket{\psi}$ and $\ket{\phi}$.

\paragraph{Constraints on $\psi$-epistemic models.---} We now show how the possible values of $\kappa(\psi, \phi)$ are constrained for ontological models that reproduce quantum predictions.
Let us start by introducing a first Lemma; its proof follows closely that of Ref.~\cite{barrett2013npe}, and is given in Appendix~\ref{app_proof_lemma1}:
\begin{lem}
\label{lem:general_ineq}
Consider $n+1$ states $\{ \psi_j \}_{0 \leq j \leq n}$ (with $n \geq 2$). For each triplet $(\psi_0, \psi_{j_1}, \psi_{j_2})$, with $1 \leq j_1 < j_2 \leq n$, consider a measurement $M_{j_1,j_2}$ with 3 possible outcomes $(m_0, m_1, m_2)$.
Then (with $j_0 = 0$)
\ba
&& \hspace{-5mm} \sum_{1 \leq j \leq n} \omega_C(\mu_{\psi_0}, \mu_{\psi_j}) \ = \sum_{1 \leq j \leq n} \kappa(\psi_0, \psi_j) \ \omega_Q(\ket{\psi_0}, \ket{\psi_j}) \nonumber \\
 && \hspace{10mm} \leq \ 1 \ + \hspace{-1mm} \sum_{1 \leq j_1 < j_2 \leq n} \, \sum_{i=0}^2 \, P_{M_{j_1,j_2}}(m_i | \psi_{j_i}). \ \label{eq:bnd_sum_wc}
\ea
\end{lem}

An interesting situation is when for each triplet of states $(\psi_0, \psi_{j_1}, \psi_{j_2})$, there exists a measurement $M_{j_1,j_2}$ such that $P_{M_{j_1,j_2}}(m_0 | \psi_0) = P_{M_{j_1,j_2}}(m_1 | \psi_{j_1}) = P_{M_{j_1,j_2}}(m_2 | \psi_{j_2}) = 0$---in which case the quantum states $(\ket{\psi_0}, \ket{\psi_{j_1}}, \ket{\psi_{j_2}})$ are said to be \emph{PP-incompatible}~\cite{caves2002cfc}.
In such a situation the left-hand side of Inequality~\eqref{eq:bnd_sum_wc} is simply 1. If the states $\ket{\psi_j}$ are furthermore such that all quantum overlaps $\omega_Q(\ket{\psi_0}, \ket{\psi_j})$ are equal (i.e., $\omega_Q(\ket{\psi_0}, \ket{\psi_j}) = \omega_Q(\ket{\psi_0}, \ket{\psi_1})$ for all $1 \leq j \leq n$), then Lemma~\ref{lem:general_ineq} implies
\ba
\frac{1}{n} \sum_{1 \leq j \leq n} \! \kappa(\psi_0, \psi_j) \ \leq \ \frac{1}{n \, \omega_Q(\ket{\psi_0}, \ket{\psi_1})} \, , \label{eq:bnd_avg_kappa_0}
\ea
which in turn implies that at least one term $\kappa(\psi_0, \psi_j)$ in the average above is upper-bounded by $1 / [n \, \omega_Q(\ket{\psi_0}, \ket{\psi_1})]$.
Refs.~\cite{barrett2013npe} and~\cite{leifer2014pem} exhibited states for which the above upper bound on the ratios of classical-to-quantum overlaps must decrease like $1/d$ and exponentially, respectively, with the Hilbert space dimension $d$ (see Appendix~\ref{app_comparison}).
We significantly improve these results by showing here the existence of states, in any dimension $d \geq 4$, with arbitrarily low values of $\kappa$.
For that we shall make use of the following Lemma:
\begin{lem}
\label{lem:existence_states}
In any Hilbert space ${\cal H}$ of dimension $d \geq 3$, and for any $n \geq 2$, there exist $n+1$ quantum states $\{ \ket{\psi_j^{(n)}} \}_{0 \leq j \leq n}$ such that
\begin{itemize}
\item[(i)] for all $1 \leq j \leq n$, $\omega_Q(\ket{\psi_0^{(n)}}, \ket{\psi_j^{(n)}}) = 1-\sqrt{1-\frac{1}{4} n^{-1/(d-2)}} \, > \, \frac{1}{8} n^{-1/(d-2)}$;
\item[(ii)] all triplets of states $( \ket{\psi_0^{(n)}}, \ket{\psi_{j_1}^{(n)}}, \ket{\psi_{j_2}^{(n)}} )$, for $1 \leq j_1 < j_2 \leq n$, are PP-incompatible.
\end{itemize}
\end{lem}

\begin{proof}
Let us fix an arbitrary state $\ket{\psi_0^{(n)}} \in {\cal H}$.
From the study of optimal Grassmannian line packings in the $(d{-}1)$-dimensional subspace orthogonal to $\ket{\psi_0^{(n)}}$, one can show that it is possible to find $n$ states $\{ \ket{\phi_j^{(n)}} \}_{1 \leq j \leq n}$ in this orthogonal subspace such that for all $j_1 \neq j_2$, $|\braket{\phi_{j_1}^{(n)}}{\phi_{j_2}^{(n)}}|^2 \leq 1 - n^{-1/(d-2)}$~\cite{love2003gra}.
From these $n$ states $\ket{\phi_j^{(n)}}$ we define, for $1 \leq j \leq n$ and with $\chi = \frac{1}{4} n^{-1/(d-2)}$ (such that $0 < \chi < \frac{1}{4}$), the states $\ket{\psi_j^{(n)}} = \sqrt{\chi} \, \ket{\psi_0^{(n)}} + \sqrt{1 {-} \chi} \, \ket{\phi_j^{(n)}}$, which indeed satisfy \emph{(i)}.

For $1 \leq j_1 < j_2 \leq n$, defining $x_1 = |\braket{\psi_0^{(n)}}{\psi_{j_1}^{(n)}}|^2$, $x_2 = |\braket{\psi_0^{(n)}}{\psi_{j_2}^{(n)}}|^2$ and $x_3 = |\braket{\psi_{j_1}^{(n)}}{\psi_{j_2}^{(n)}}|^2$, one has $x_1 = x_2 = \chi$ and $x_3 = \big| \chi + (1{-}\chi) \braket{\phi_{j_1}^{(n)}}{\phi_{j_2}^{(n)}} \big|^2 \leq  \big[ \chi + (1{-}\chi) |\braket{\phi_{j_1}^{(n)}}{\phi_{j_2}^{(n)}}| \big]^2 \leq  \big[ \chi + (1{-}\chi) \sqrt{1 - n^{-1/(d-2)}} \big]^2 \leq \big( 1{-}2\chi \big)^2$. One then finds that $x_1 + x_2 + x_3 \leq 2\chi + (1{-}2\chi)^2 < 1$ and $(1{-}x_1{-}x_2{-}x_3)^2 \geq 4\chi^2 (1{-}2\chi)^2 \geq 4 x_1 x_2 x_3$, which shows that $( \ket{\psi_0^{(n)}}, \ket{\psi_{j_1}^{(n)}}, \ket{\psi_{j_2}^{(n)}} )$ are PP-incompatible~\cite{caves2002cfc,barrett2013npe}.
\end{proof}

The states just constructed thus satisfy the assumptions of Eq.~\eqref{eq:bnd_avg_kappa_0}, which gives
\ba
\frac{1}{n} \sum_{j} \kappa(\psi_0^{(n)}\!, \psi_j^{(n)}) \leq \frac{1}{n \Big( 1{-}\sqrt{1{-}\frac{1}{4} n^{-1/(d-2)}} \Big)} < 8/n^{\frac{d-3}{d-2}} \, . \nonumber \\[-3mm] \label{eq:bnd_avg_8n}
\ea
In particular, for each $n$, there exists at least one state $\psi_j^{(n)}$ such that $\kappa(\psi_0^{(n)}, \psi_j^{(n)}) \leq 8/n^{\frac{d-3}{d-2}}$. Noting that for any (fixed) $d \geq 4$, this upper-bound tends to 0 as $n \to \infty$, we conclude that, as claimed before:
\begin{thm}
\label{thm:main_result_d4}
For any ontological model reproducing quantum measurement statistics in a Hilbert space of dimension $d \geq 4$, there exist (nonorthogonal) states $\phi, \psi$ with an arbitrarily low ratio of classical-to-quantum overlaps $\kappa(\psi,\phi)$.
\end{thm}

It remains an open question whether the same result holds for a Hilbert space of dimension $d = 3$ or not. On the other hand, in the case of dimension $d=2$, one can check that the Kochen-Specker model~\cite{kochen1967the,harrigan2010ein} for projective measurements is maximally $\psi$-epistemic---i.e., is such that $\kappa(\psi,\phi) = 1$ for all nonorthogonal states $\psi, \phi$.

\paragraph{Upper-bounding $\kappa(\psi,\phi)$ experimentally. ---}

The argument leading to Theorem~\ref{thm:main_result_d4} above required us to consider infinitely many states ($n{+}1 \to \infty$) and measurements (in $\frac{n(n-1)}{2} \to \infty$ bases).
Nevertheless, using finitely many states and measurements already allows one to put nontrivial bounds on some ratios of classical-to-quantum overlaps $\kappa(\psi,\phi)$, which can be verified experimentally.

In order to simplify the discussion below, we will now assume (as in Ref.~\cite{barrett2013npe}) that in the ontological model under study, $\kappa(\psi,\phi)$ is lower-bounded by a fixed value $\kappa_0$: $\kappa(\psi,\phi) \geq \kappa_0$ for all $\psi,\phi$.
Our goal will then be to upper-bound the value of $\kappa_0$. Note, however, that in all we write below, $\kappa_0$ can simply be replaced by $\min_{1 \leq j \leq n} [\kappa(\psi_0, \psi_j)]$ for the states under consideration---or even, when all quantum overlaps $\omega_Q(\ket{\psi_0}, \ket{\psi_j})$ are equal, by the average $\frac{1}{n} \sum_{j} \kappa(\psi_0, \psi_j)$ (as in Eq.~\eqref{eq:bnd_avg_kappa_0}).

In regard to experimental tests, one may want to minimize the number of states and measurements to be used, as well as the bound on $\kappa_0$.
Although the PP-incompatibility property used above for each triplet $(\ket{\psi_0^{(n)}}, \ket{\psi_{j_1}^{(n)}}, \ket{\psi_{j_2}^{(n)}})$ was quite convenient to obtain our theoretical results, using states with such a property is in fact not optimal for experimental purposes.
Note indeed that this PP-incompatibility criterion is not necessary in our approach, and neither is the assumption that all quantum overlaps $\omega_Q(\ket{\psi_0}, \ket{\psi_j})$ are equal: it is still possible without those to bound the value of $\kappa_0$ by using Eq.~\eqref{eq:bnd_sum_wc} of Lemma~\ref{lem:general_ineq}, which implies that
\ba
\kappa_0 \ \leq \ \frac{1 + \sum_{j_1 < j_2} \, \sum_{i=0}^2 \, P_{M_{j_1,j_2}}(m_i | \psi_{j_i})}{\sum_{j} \omega_Q(\ket{\psi_0}, \ket{\psi_j})}. \
\label{eq:bnd_min_k}
\ea
The right-hand side above can be obtained experimentally by preparing the $n+1$ quantum states $\ket{\psi_j}$ and measuring them according to some measurements $M_{j_1,j_2}$, so as to estimate the probabilities $P_{M_{j_1,j_2}}(m_i | \psi_{j_i})$. Note that no assumption needs to be made on the measurements to calculate the bound (they can in principle correspond to any unknown POVM); one, however, needs to make sure that the states $\ket{\psi_j}$ are reliably prepared, so that the quantum overlaps $\omega_Q(\ket{\psi_0}, \ket{\psi_j})$ are reliably calculated.

Already for a Hilbert space of dimension $d = 3$, it is possible to derive nontrivial bounds for $\kappa_0$.
The authors of Ref.~\cite{barrett2013npe} constructed a set of $n+1 = 9+1$ states $\{ \ket{\psi_j} \}$ which allowed them to show that $\kappa_0 \leq 0.95$. We found that using $n+1 = 3+1$ states is actually sufficient in principle to get a nontrivial bound, although the bound we obtained is quite close to 1, and may be hard to demonstrate experimentally: for the states and measurements we found, one gets, from~\eqref{eq:bnd_min_k}, $\kappa_0 \leq 0.9964$ (see Appendix~\ref{app_explicit_states_meas}).
Moving to $n + 1 = 4 + 1$, we obtained, with the states and measurements specified in Appendix~\ref{app_explicit_states_meas}, the bound $\kappa_0 \leq 0.9361$---which already improves (using fewer states and measurements, i.e. fewer experimental resources) the bound given in Ref.~\cite{barrett2013npe}, thus allowing one to rule out a strictly larger class of $\psi$-epistemic models for quantum measurement statistics.

An important consideration for experimental tests is the resistance to noise. In order to get a rough estimate of how robust our results are, let us assume that the probabilities $P_{M_{j_1,j_2}}(m_i | \psi_{j_i})$ in Eq.~\eqref{eq:bnd_min_k} are each estimated up to a quantity $\epsilon$ (taken for simplicity to be the same for all these probabilities). In order to take the worst case into account in the bound of Eq.~\eqref{eq:bnd_min_k}, we thus add $\epsilon$ to each term $P_{M_{j_1,j_2}}(m_i | \psi_{j_i})$---that is, we add $3 \!\times\! \frac{n(n-1)}{2} \, \epsilon$ in total to the numerator of the fraction in~\eqref{eq:bnd_min_k}.
The bound remains nontrival if it is lower than 1; for the states and measurements specified in Appendix~\ref{app_explicit_states_meas}, this requires $\epsilon < 5 \times 10^{-3}$, which is also more robust than what was found in Ref.~\cite{barrett2013npe}.

Increasing the number of states (and measurement bases) allows one to decrease the bound on $\kappa_0$. Numerically we could obtain a bound $\kappa_0 \leq 0.6408$ using $n+1 = 20+1$ states~\cite{branciard2014how}. We expect it is possible to set an even lower bound by using more states; how low the bound can be, and in particular whether it can be taken down to zero (for $n \to \infty$, as in the case $d \geq 4$), is an open question. Note also that as $n$ becomes large, the robustness to noise in general decreases~\cite{footnote_robust_noise}.

In dimension $d = 4$, we found again that $n+1 = 3+1$ is enough to set a nontrivial bound on $\kappa_0$ (we actually found the same bound as in the $d = 3$ case above). With $n+1 = 4+1$ states, we could find an improvement over the $d = 3$ case, obtaining the bound $\kappa_0 \leq 0.9054$ (Appendix~\ref{app_explicit_states_meas}). Regarding the robustness to noise, the states and measurements we use require a value of $\epsilon < 7 \times 10^{-3}$ to give a nontrivial bound on $\kappa_0$.
We recall that, from the result of Theorem~\ref{thm:main_result_d4}, using more states allows one to set an arbitrarily low upper bound on $\kappa_0$. Namely, any bound $\kappa_0 \leq \bar\kappa$, for any $\bar\kappa > 0$, can be obtained from~\eqref{eq:bnd_avg_8n} by using $n \geq (8/\bar\kappa)^{\frac{d-2}{d-3}}$ states (for $d \geq 4$) constructed as in Lemma~\ref{lem:existence_states} (which may, however, not be optimal).
Again, as the number of states gets larger, the resistance to noise decreases---e.g., the states used in~\eqref{eq:bnd_avg_8n} require a noise parameter $\epsilon \lesssim 1/ \big( 12 \, n^\frac{d-1}{d-2} \big)$ for the bound to be lower than~1.

\paragraph{Discussion. ---}

We have proven in this Letter that for any ontological model reproducing the quantum measurement statistics in a Hilbert space of dimension $d \geq 4$, there exist states $\psi, \phi$ whose ratio of classical-to-quantum overlaps is arbitrarily small. Recalling that this ratio quantifies how much the model explains of the indistinguishability of the quantum states $\ket{\psi}$ and $\ket{\phi}$, we conclude that $\psi$-epistemic models must be arbitrarily bad in general at explaining quantum indistinguishability.

Note that the states $\ket{\psi_j^{(n)}}$ ($1 \leq j \leq n$) used in our proof become closer and closer, as $n$ increases, to being orthogonal to $\ket{\psi_0^{(n)}}$.
Recalling that $|\braket{\psi_0^{(n)}}{\psi_j^{(n)}}|^2 = \frac{1}{4} n^{-1/(d-2)}$, what is actually shown by Eq.~\eqref{eq:bnd_avg_8n} is the existence of states $\psi, \phi$ such that
\ba
\kappa(\psi, \phi) \ \leq  \ \frac{4^d}{8} |\braket{\psi}{\phi}|^{2(d-3)} \ \ \text{with} \ \ |\braket{\psi}{\phi}| \to 0. \quad
\ea
It has been argued, for reasons to do with coarse graining of measurements, that it may be natural to expect $\kappa(\psi, \phi) \to 0$ when $|\braket{\psi}{\phi}| \to 0$~\cite{leifer2014pem,leifer_private}. Our result thus imposes some constraint on the possible scaling of $\kappa(\psi, \phi)$ as the two states tend to orthogonal states, for any $\psi$-epistemic model reproducing quantum measurement statistics---e.g., those of Refs~\cite{lewis2012dis,aaronson2013psi}.
Another question worth investigating would be to see what bounds on $\kappa(\psi,\phi)$ can be derived for some fixed values of $|\braket{\psi}{\phi}|$.

Interestingly, our theoretical result here (as well as that of Ref.~\cite{pusey2012ont} and many of the other no-go theorems that followed) does not depend on the specific form of the Born rule.
By considering PP-incompatible triplets of states, we just used in our proof the fact that when quantum theory assigns zero probability to a measurement outcome, the ontological model must do the same~\cite{footnote_born_rule}.
It would be interesting to see if, when the Born rule predicts a nonzero probability, its specific form imposes stronger constraints on possible $\psi$-epistemic models for quantum measurement statistics (e.g., if it allows one to prove our result for $d=3$ as well).
This may give more insights on which types of models can reproduce quantum theory.

Using finitely many states and measurements, one can only set finite bounds on the ratios of classical-to-quantum overlaps. We exhibited states and measurements allowing one to put nontrivial bounds already for $d = 3$.
These can be used experimentally to rule out some families of $\psi$-epistemic models---namely, those which assign larger ratios of classical-to-quantum overlaps than the bound demonstrated, for the states used experimentally (recall that $\kappa_0$ used for simplicity above can be replaced by $\min_{1 \leq j \leq n} [\kappa(\psi_0, \psi_j)]$).
In particular, as soon as the bound is lower than 1, maximally $\psi$-epistemic models are ruled out. It is worth emphasizing that this can, in principle, be done without introducing any auxiliary assumption, like the independence condition used, e.g., in the experimental test reported in Ref.~\cite{nigg2012can}.
Furthermore, the allowed amount of noise for the states and measurements exhibited in this Letter for dimensions $d = 3$ and 4, should make an experiment feasible with current technology, e.g., in photonics.

\medskip

\paragraph{Acknowledgments. ---}

I acknowledge financial support from a Discovery Early Career Researcher Award (DE140100489) of the Australian Research Council and from the `Retour Post-Doctorants' program (ANR-13-PDOC-0026) of the French National Research Agency.

\appendix

\setcounter{secnumdepth}{3}
\renewcommand{\thesubsection}{\Alph{subsection}}

\renewcommand{\theequation}{A\arabic{equation}}
\setcounter{equation}{0}

\subsection*{Appendix}

\subsection{Proof of Lemma~1}
\label{app_proof_lemma1}

The proof of Lemma~1 is quite similar to that given in Appendix~1 of Ref.~\cite{barrett2013npe}. For completeness and clarity, we present it here explicitly.

\medskip

For any epistemic state $\mu_\psi$ (a nonnegative, normalised function over the ontic state space $\Lambda$), let us define the set
\ba
\Psi = \{ (\lambda, x) \in \Lambda \times \mathbb{R}^+ \, | \ 0 \leq x \leq \mu_\psi(\lambda) \}.
\ea
Denoting by $\text{d} \nu = \text{d} \lambda \times \text{d} x$ the volume measure on the space $\Lambda \times \mathbb{R}^+$, where $\text{d} x$ is the Lebesgue measure of $\mathbb{R}^+$, the volume of the set $\Psi$ is, by normalisation of $\mu_\psi$,
\ba
\nu(\Psi) = \int_\Psi \text{d} \nu = \int_\Lambda \Big( \int_0^{\mu_\psi(\lambda)} \text{d} x \Big) \, \text{d} \lambda  = \int_\Lambda \mu_\psi(\lambda) \, \text{d} \lambda = 1. \nonumber \\[-1mm]
\ea

For two epistemic states $\mu_{\psi_0}$ and $\mu_{\psi_j}$, one has
\ba
\Psi_0 \cap \Psi_j = \{ (\lambda, x) \, | \ 0 \leq x \leq \min[ \mu_{\psi_0}(\lambda), \mu_{\psi_j}(\lambda) ] \}, \quad \quad
\ea
and hence
\ba
\nu(\Psi_0 \cap \Psi_j) = \int_\Lambda \min[ \mu_{\psi_0}(\lambda), \mu_{\psi_j}(\lambda) ] \, \text{d} \lambda = \omega_C(\mu_{\psi_0}, \mu_{\psi_j}). \nonumber \\[-1mm] \label{eq:app_gamma2}
\ea

For now three epistemic states $\mu_{\psi_0}$, $\mu_{\psi_{j_1}}$ and $\mu_{\psi_{j_2}}$ one gets, as above,
\ba
\nu(\Psi_0 \cap \Psi_{j_1} \cap \Psi_{j_2}) = \int_\Lambda \min[ \mu_{\psi_0}(\lambda), \mu_{\psi_{j_1}}(\lambda), \mu_{\psi_{j_2}}(\lambda) ] \, \text{d} \lambda. \nonumber \\[-1mm] \label{eq:app_gamma3}
\ea
Consider a measurement $M_{j_1,j_2}$ with 3 possible outcomes $(m_0, m_1, m_2)$. From the normalisation condition of Eq.~(2), we have that for each $\lambda$, $\sum_{i=0}^2 \, \xi_{M_{j_1,j_2}}(m_i|\lambda) = 1$.
From~\eqref{eq:app_gamma3}, it then follows that (with $j_0 = 0$)
\ba
\nu(\Psi_0 \cap \Psi_{j_1} \! \cap \Psi_{j_2}) &=& \int_{\!\Lambda} \sum_{i=0}^2 \xi_{M_{j_1,j_2}}(m_i|\lambda) \min_{i'}[ \mu_{\psi_{j_{i'}}}\!(\lambda) ] \, \text{d} \lambda \nonumber \\
&\leq& \ \sum_{i=0}^2 \, \int_\Lambda \, \xi_{M_{j_1,j_2}}(m_i|\lambda) \, \mu_{\psi_{j_{i}}}(\lambda) \, \text{d} \lambda \nonumber \\
&& \ = \ \sum_{i=0}^2 \, P_{M_{j_1,j_2}}(m_i|\psi_{j_{i}}).  \label{eq:app_gamma3bis}
\ea

\medskip

From the $n+1$ states $\{ \psi_j \}_{0 \leq j \leq n}$, let us define, for $1 \leq j \leq n$, the $n$ measurable sets
\ba
A_j = \Psi_0 \cap \Psi_j.
\ea
The second Bonferroni inequality applied to these sets writes
\ba
\nu \Big( \bigcup_j A_j \Big) \ \geq \ \sum_j \nu \big( A_j \big) \ - \ \sum_{j_1 < j_2} \nu \big( A_{j_1} \cap A_{j_2} \big). \quad \quad
\ea
Using Eqs.~\eqref{eq:app_gamma2} and~\eqref{eq:app_gamma3bis} (noting that $A_{j_1} \cap A_{j_2} = \Psi_0 \cap \Psi_{j_1} \cap \Psi_{j_2}$) and the fact that $\nu \big( \bigcup_j A_j \big) = \nu \big( \bigcup_j \Psi_0 \cap \Psi_j \big) \leq \nu \big( \Psi_0 \big) = 1$, this implies that
\ba
\sum_j \omega_C(\mu_{\psi_0}, \mu_{\psi_j}) \, \leq \, 1 + \! \sum_{j_1 < j_2} \sum_{i=0}^2 \, P_{M_{j_1,j_2}}(m_i|\psi_{j_{i}}). \quad \quad
\ea

Note that this inequality holds for any ontological model, independently of whether it reproduces quantum measurement statistics or not. Quantum theory enters when we write $\omega_C(\mu_{\psi_0}, \mu_{\psi_j}) = \kappa(\psi_0, \psi_j) \, \omega_Q(\ket{\psi_0}, \ket{\psi_j})$ in the l.h.s. of Inequality~(7) of Lemma~1.

\qed

\subsection{Comparison with the results of Refs.~\cite{barrett2013npe,leifer2014pem}}
\label{app_comparison}

Here we compare our results with the bounds obtained in Refs.~\cite{barrett2013npe,leifer2014pem} on the ratios of classical-to-quantum overlaps. For simplicity let us assume that there exists $\kappa_0$ such that $\kappa(\psi,\phi) \geq \kappa_0$ for all $\psi,\phi$, recalling that $\kappa_0$ can simply be taken to be $\min_{1 \leq j \leq n} [\kappa(\psi_0, \psi_j)]$ for the states under consideration below---or even, as all quantum overlaps $\omega_Q(\ket{\psi_0}, \ket{\psi_j})$ will here be equal, the average $\frac{1}{n} \sum_{j} \kappa(\psi_0, \psi_j)$.

\medskip

In Ref.~\cite{barrett2013npe}, the $n = d^2$ quantum states $\{ \ket{\psi_j} \}_{1 \leq j \leq d^2}$ considered in our analysis were chosen to be the states of $d$ mutually unbiased bases in a Hilbert space of a prime power dimension $d \geq 4$, while $\ket{\psi_0}$ was taken to be one of the $d$ states of the remaining mutually unbiased basis (recall that for a prime power $d$, there exist $d+1$ mutually unbiased bases~\cite{wootters1989osd}). These states satisfy $\omega_Q(\ket{\psi_0}, \ket{\psi_j}) = 1 - \sqrt{1 - 1/d}$ for all $1 \leq j \leq d^2$. Furthermore, for $1 \leq j_1 < j_2 \leq d^2$, the triplets of states $( \ket{\psi_0}, \ket{\psi_{j_1}}, \ket{\psi_{j_2}} )$ thus chosen are all PP-incompatible. As in Ref.~\cite{barrett2013npe}, one can then conclude from Eq.~(8) of the main text that for $d \geq 4$ a prime power,
\ba
\kappa_0 \ \leq \ \frac{1 }{d^2 (1 - \sqrt{1 - 1/d})} \ < \ \frac{2}{d}.
\ea
This was furthermore shown to also imply a bound arbitrary for an dimension $d$, namely $\kappa_0 < 4/(d-1)$~\cite{barrett2013npe}.

\medskip

In Ref.~\cite{leifer2014pem}, Leifer followed a different approach to also set bounds on the ratios of classical-to-quantum overlaps. Nevertheless, the states he considered can also be used with our approach. Namely, for a Hilbert space of any dimension $d \geq 3$, let us define the state $\ket{\psi_0} = (1, 0, ..., 0)^T$, and let the $n = 2^{d-1}$ states $\{ \ket{\psi_j} \}_{1 \leq j \leq 2^{d-1}}$ be the Hadamard states $\frac{1}{\sqrt{d}}(1, \pm 1, ..., \pm 1)^T$. One again has in this case $\omega_Q(\ket{\psi_0}, \ket{\psi_j}) = 1 - \sqrt{1 - 1/d}$ for all $1 \leq j \leq 2^{d-1}$. Furthermore, noting that $|\braket{\psi_0}{\psi_{j_1}}|^2 = |\braket{\psi_0}{\psi_{j_2}}|^2 = 1/d$ and $|\braket{\psi_{j_1}}{\psi_{j_2}}|^2 \leq (1-2/d)^2$, one can check (using the criterion of Refs.~\cite{caves2002cfc,barrett2013npe}) that all triplets of states $( \ket{\psi_0}, \ket{\psi_{j_1}}, \ket{\psi_{j_2}} )$, for $1 \leq j_1 < j_2 \leq 2^{d-1}$, are still PP-incompatible. Eq.~(8) then leads to a bound on $\kappa_0$ that decreases exponentially with the dimension $d$:
\ba
\kappa_0 \ \leq \ \frac{1 }{2^{d-1} (1 - \sqrt{1 - 1/d})} \ < \ \frac{4d}{2^d}. \label{eq:app_k0_exp}
\ea
This bound actually scales better than that of Ref.~\cite{leifer2014pem} ($\kappa_0 \leq 2d (1{-}\epsilon/2)^d$ for dimensions $d$ divisible by 4, where $\epsilon$ is a positive constant smaller than 1).

Note that the exponential scaling just obtained requires one to consider an exponential number of states.
If one also considers an exponential number of states---say, $n = 2^m$ for some $m \geq 1$---in Lemma~2 of the main text, one gets from Eq.~(9) a bound $\kappa_0 < 8 / 2^{\frac{d-3}{d-2}m}$ which also gives and exponential scaling (as $m$ increases), already for $d = 4$.
(For the same number of states as above, i.e. $m=d-1$, the states of Lemma~2 give yet a slightly better scaling with dimension than~\eqref{eq:app_k0_exp}.)
From our study it appears that the scaling of the bound on $\kappa_0$ as the number of states increases, or the scaling of $\kappa(\psi, \phi)$ as the inner product $|\braket{\psi}{\phi}|$ tends to 0 (see Discussion in the main text), are more relevant than their scaling with the dimension of the Hilbert space.

\subsection{Explicit states and measurements for experimental tests \\ (for $d = 3,4$ and $n = 3,4$)}
\label{app_explicit_states_meas}

We give here the best states and measurements we found for bounding the value of $\kappa_0$, in dimension $d = 3$ using $n+1 = 3+1$ and $n+1=4+1$ states, and in dimension $d = 4$ using $n+1=4+1$ states (for the case $n+1 = 3+1$ in dimension $d=4$, the best states and measurements we found are the same as in the qutrit case).
The Supplemental Material of Ref.~\cite{branciard2014how} also contains the best states and measurements we found for the case $d=3, n+1 = 20+1$, which give a bound $\kappa_0 \leq 0.6408$ and a robustness to noise $\epsilon < 10^{-3}$. Note that the states and measurements we obtained are the result of a numerical search, and are not claimed to be optimal.

\medskip

The states are written below as ket vectors, while the measurements are given in the form of unitary matrices representing a projection eigenbasis: each column of $M_{j_1,j_2}$ labelled by $i = 0, ..., d-1$ corresponds to the eigenstate $\ket{m_i^{j_1,j_2}}$, giving outcome $m_i$. Note that for the case $d = 4$, the last column ($i=3$) of each matrix $M_{j_1,j_2}$ is always orthogonal to the three states $\ket{\psi_0}$, $\ket{\psi_{j_1}}$ and $\ket{\psi_{j_2}}$, and hence the corresponding outcome never appears when the measurement is performed on one of these states. To rigorously define a 3-outcome measurement, we can arbitrarily set for instance $m_3 = m_0$, so that the three measurement operators are $\{ \ketbra{m_0^{j_1,j_2}} + \ketbra{m_3^{j_1,j_2}} , \ketbra{m_1^{j_1,j_2}} , \ketbra{m_2^{j_1,j_2}} \}$.

\subsubsection*{--- $\, d = 3, \, n = 3 \,$ ---}

As claimed in the main text, using $n+1 = 3+1$ states is sufficient to get nontrivial bounds on $\kappa_0$.

We shall use here the shorthand notations $c_k = \cos \theta_k$ and $s_k = \sin \theta_k$.
The best states we found are given by
\ba
& \ket{\psi_0} =
\left( \begin{array}{c}
1 \\ 0 \\ 0
\end{array} \right) \! , \ \
\ket{\psi_1} =
\left( \begin{array}{c}
c_1 \\ s_1 \\ 0
\end{array} \right) \! , \nonumber \\[1mm]
& \ket{\psi_2} =
\left( \begin{array}{c}
c_2 \\ s_2 c_3 \\ s_2 s_3
\end{array} \right) \! , \ \
\ket{\psi_3} =
\left( \begin{array}{c}
c_2 \\ s_2 c_3 \\ - s_2 s_3
\end{array} \right) \! , \nonumber
\ea
with $\theta_1 = 1.1945$, $\theta_2 = 0.2839$ and $\theta_3 = 1.8423$. For the following measurement bases:
\ba
& M_{1,2} =
\left( \begin{array}{ccc}
c_4 & s_4 c_6 & s_4 s_6 \\
s_4 c_5 & \ - c_4 c_5 c_6 - s_5 s_6 \ & - c_4 c_5 s_6 + s_5 c_6 \\
- s_4 s_5 & c_4 s_5 c_6 - c_5 s_6 & c_4 s_5 s_6 + c_5 c_6
\end{array} \right) \! , \nonumber \\[1mm]
& M_{1,3} =
\left( \begin{array}{ccc}
c_4 & s_4 c_6 & s_4 s_6 \\
s_4 c_5 & \ - c_4 c_5 c_6 - s_5 s_6 \ & - c_4 c_5 s_6 + s_5 c_6 \\
s_4 s_5 & - c_4 s_5 c_6 + c_5 s_6 & - c_4 s_5 s_6 - c_5 c_6
\end{array} \right) \! , \nonumber \\[1mm]
& M_{2,3} =
\left( \begin{array}{ccc}
c_7 & s_7/\sqrt{2} & s_7/\sqrt{2} \\
- s_7 & c_7/\sqrt{2} & c_7/\sqrt{2} \\
0 & \, -1/\sqrt{2} \, \ & \,1\,/\sqrt{2}
\end{array} \right) \! , \nonumber
\ea
with $\theta_4 = 1.6276$, $\theta_5 = 2.2192$, $\theta_6 = 0.3100$ and $\theta_7 = 1.4269$, the bound obtained from Eq.~(10) of the main text is found to be $\kappa_0 \leq 0.9964$.

When the probabilities appearing in Eq.~(10) are estimated up to a quantity $\epsilon$, as considered in the main text, the upper bound remains non-trivial in the worst case as long as $\frac{1 + \sum_{j_1 < j_2} \, \sum_{i=0}^2 \, P_{M_{j_1,j_2}}(m_i | \psi_{j_i}) + \frac{3}{2}n(n-1) \epsilon}{\sum_{j} \omega_Q(\ket{\psi_0}, \ket{\psi_j})} < 1$, i.e. when
\ba
 \epsilon < \frac{\sum_{j} \omega_Q(\ket{\psi_0}, \ket{\psi_j}) - 1 - \sum_{j_1 < j_2} \sum_{i=0}^2 P_{M_{j_1,j_2}\!}(m_i | \psi_{j_i})}{\frac{3}{2}n(n-1)} . \nonumber \\[-1mm]
\ea
For the above states and measurements, we obtain the requirement that $\epsilon < 6 \times 10^{-4}$.

\subsubsection*{--- $\, d = 3, \, n = 4 \,$ ---}

The following $n+1 = 4+1$ states:
\ba
& \ket{\psi_0} =
\left( \begin{array}{c}
1 \\ 0 \\ 0
\end{array} \right) \! , \ \
\ket{\psi_1} =
\left( \begin{array}{c}
c_\theta \\ s_\theta \\ 0
\end{array} \right) \! , \ \
\ket{\psi_2} =
\left( \begin{array}{c}
c_\theta \\ 0 \\ s_\theta
\end{array} \right) \! , \nonumber \\[1mm]
& \ket{\psi_3} =
\left( \begin{array}{c}
c_\theta \\ -s_\theta \\ 0
\end{array} \right) \! , \ \
\ket{\psi_4} =
\left( \begin{array}{c}
c_\theta \\ 0 \\ -s_\theta
\end{array} \right) \! , \nonumber
\ea
with $c_\theta = \cos \theta$, $s_\theta = \sin \theta$ and $\theta = 0.7152$,
and the following measurement bases:
\ba
& M_{1,2} =
\left( \begin{array}{ccc}
c_\varphi & s_\varphi / \sqrt{2} & s_\varphi / \sqrt{2} \\
s_\varphi / \sqrt{2} & \ -(1{+}c_\varphi) / 2 \ & (1{-}c_\varphi) / 2 \\
s_\varphi / \sqrt{2} & (1{-}c_\varphi) / 2 & -(1{+}c_\varphi) / 2
\end{array} \right) \! , \nonumber \\[1mm]
& M_{1,3} =
\left( \begin{array}{ccc}
0 & 1 / \sqrt{2} & 1 / \sqrt{2} \\
0 & \ -1 / \sqrt{2} \ & 1 / \sqrt{2} \\
1 & 0 & 0
\end{array} \right) \! , \nonumber \\[1mm]
& M_{1,4} =
\left( \begin{array}{ccc}
c_\varphi & s_\varphi / \sqrt{2} & s_\varphi / \sqrt{2} \\
s_\varphi / \sqrt{2} & \ -(1{+}c_\varphi) / 2 \ & (1{-}c_\varphi) / 2 \\
-s_\varphi / \sqrt{2} & -(1{-}c_\varphi) / 2 & (1{+}c_\varphi) / 2
\end{array} \right) \! , \nonumber \\[1mm]
& M_{2,3} =
\left( \begin{array}{ccc}
c_\varphi & s_\varphi / \sqrt{2} & s_\varphi / \sqrt{2} \\
-s_\varphi / \sqrt{2} & \ -(1{-}c_\varphi) / 2 \ & (1{+}c_\varphi) / 2 \\
s_\varphi / \sqrt{2} & -(1{+}c_\varphi) / 2 & (1{-}c_\varphi) / 2
\end{array} \right) \! , \nonumber \\[1mm]
& M_{2,4} =
\left( \begin{array}{ccc}
0 & 1 / \sqrt{2} & 1 / \sqrt{2} \\
1 & 0 & 0 \\
0 & \ -1 / \sqrt{2} \ & 1 / \sqrt{2}
\end{array} \right) \! , \nonumber \\[1mm]
& M_{3,4} =
\left( \begin{array}{ccc}
c_\varphi & s_\varphi / \sqrt{2} & s_\varphi / \sqrt{2} \\
-s_\varphi / \sqrt{2} & (1{+}c_\varphi) / 2 & -(1{-}c_\varphi) / 2 \\
-s_\varphi / \sqrt{2} & \ -(1{-}c_\varphi) / 2 \ & (1{+}c_\varphi) / 2
\end{array} \right) \! ,  \nonumber
\ea
with $c_\varphi = \cos \varphi$, $s_\varphi = \sin \varphi$ and $\varphi = 1.4436$, give, from Eq.~(10), the bound $\kappa_0 \leq 0.9361$.
Note that as the states $\ket{\psi_j}$ ($1 \leq j \leq 4$) considered here all have the same quantum overlap with $\ket{\psi_0}$, this bound actually holds for the average ratio of classical-to-quantum overlaps, $\frac{1}{4} \sum_{j} \kappa(\psi_0, \psi_j)$ (see main text).

Regarding the robustness to experimental noise, the requirement on the quantity $\epsilon$, calculated as above, is found to be $\epsilon < 5 \times 10^{-3}$, i.e. $\epsilon$ can be an order of magnitude larger than in the previous $n=3$ case.
(Note, besides, that the states and measurements that give the highest robustness to noise may in general be slightly different from those that optimize the bound on $\kappa_0$ in the noise-free case.)

\medskip

\subsubsection*{--- $\, d = 4, \, n = 4 \,$ ---}

The following $n+1 = 4+1$ states:
\ba
& \ket{\psi_0} =
\left( \begin{array}{c}
1 \\ 0 \\ 0 \\ 0
\end{array} \right) \! , \
\ket{\psi_1} =
\left( \! \begin{array}{c}
c_\theta \\ s_\theta / \sqrt{3} \\ s_\theta / \sqrt{3} \\ s_\theta / \sqrt{3}
\end{array} \right) \! , \
\ket{\psi_2} =
\left( \! \begin{array}{c}
c_\theta \\ s_\theta / \sqrt{3} \\ -s_\theta / \sqrt{3} \\ -s_\theta / \sqrt{3}
\end{array} \right) \! , \nonumber \\[1mm]
& \ket{\psi_3} =
\left( \! \begin{array}{c}
c_\theta \\ -s_\theta / \sqrt{3} \\ s_\theta / \sqrt{3} \\ -s_\theta / \sqrt{3}
\end{array} \right) \! , \
\ket{\psi_4} =
\left( \! \begin{array}{c}
c_\theta \\ -s_\theta / \sqrt{3} \\ -s_\theta / \sqrt{3} \\ s_\theta / \sqrt{3}
\end{array} \right) \! , \nonumber
\ea
with now $\theta = 0.7274$,
and the following measurement bases:
\ba
& M_{1,2} =
\left( \begin{array}{cccc}
c_\varphi & s_\varphi / \sqrt{2} & s_\varphi / \sqrt{2} & 0 \\
s_\varphi & -c_\varphi / \sqrt{2} & -c_\varphi / \sqrt{2} & 0 \\
0 & -1/2 & 1/2 & 1/\sqrt{2} \\
0 & -1/2 & 1/2 & -1/\sqrt{2}
\end{array} \right) \! , \nonumber \\[1mm]
& M_{1,3} =
\left( \begin{array}{cccc}
c_\varphi & s_\varphi / \sqrt{2} & s_\varphi / \sqrt{2} & 0 \\
0 & -1/2 & 1/2 & 1/\sqrt{2} \\
s_\varphi & -c_\varphi / \sqrt{2} & -c_\varphi / \sqrt{2} & 0 \\
0 & -1/2 & 1/2 & -1/\sqrt{2}
\end{array} \right) \! , \nonumber \\[1mm]
& M_{1,4} =
\left( \begin{array}{cccc}
c_\varphi & s_\varphi / \sqrt{2} & s_\varphi / \sqrt{2} & 0 \\
0 & -1/2 & 1/2 & 1/\sqrt{2} \\
0 & -1/2 & 1/2 & -1/\sqrt{2} \\
s_\varphi & -c_\varphi / \sqrt{2} & -c_\varphi / \sqrt{2} & 0
\end{array} \right) \! , \nonumber \\[1mm]
& M_{2,3} =
\left( \begin{array}{cccc}
c_\varphi & s_\varphi / \sqrt{2} & s_\varphi / \sqrt{2} & 0 \\
0 & -1/2 & 1/2 & 1/\sqrt{2} \\
0 & 1/2 & -1/2 & 1/\sqrt{2} \\
-s_\varphi & c_\varphi / \sqrt{2} & c_\varphi / \sqrt{2} & 0
\end{array} \right) \! , \nonumber \\[1mm]
& M_{2,4} =
\left( \begin{array}{cccc}
c_\varphi & s_\varphi / \sqrt{2} & s_\varphi / \sqrt{2} & 0 \\
0 & -1/2 & 1/2 & 1/\sqrt{2} \\
-s_\varphi & c_\varphi / \sqrt{2} & c_\varphi / \sqrt{2} & 0 \\
0 & 1/2 & -1/2 & 1/\sqrt{2}
\end{array} \right) \! , \nonumber \\[1mm]
& M_{3,4} =
\left( \begin{array}{cccc}
c_\varphi & s_\varphi / \sqrt{2} & s_\varphi / \sqrt{2} & 0 \\
-s_\varphi & c_\varphi / \sqrt{2} & c_\varphi / \sqrt{2} & 0 \\
0 & -1/2 & 1/2 & 1/\sqrt{2} \\
0 & 1/2 & -1/2 & 1/\sqrt{2}
\end{array} \right) \! , \nonumber
\label{eq:def_meas_d4n4}
\ea
with now $\varphi = 1.4946$, give, from Eq.~(10), the bound $\kappa_0 \leq 0.9054$.
Again, as the states $\ket{\psi_j}$ ($1 \leq j \leq 4$) considered here all have the same quantum overlap with $\ket{\psi_0}$, this bound actually holds for the average $\frac{1}{4} \sum_{j} \kappa(\psi_0, \psi_j)$.

The requirement on the noise parameter $\epsilon$ is found here to be $\epsilon < 7 \times 10^{-3}$; for $d=4$, an experiment with the above $n+1=4+1$ states would thus allow for some larger noise than one with the previous qutrit states.


\begin{thebibliography}{23}
\bibitem{footnote_robust_noise} Intuitively: the number of terms in the denominator of Eq.~\eqref{eq:bnd_min_k} increases as $O(n)$, while the number of probabilities (and hence of possible error terms) in the numerator increases as $O(n^2)$.
\bibitem{footnote_born_rule} The fact that the model must reproduce all quantum predictions, more generally, was only used to prove that $\omega_C(\mu_\psi, \mu_\phi) \leq \omega_Q(\ket{\psi}, \ket{\phi})$, and to justify the relevance of comparing the classical and quantum overlaps, through the definition of their ratio $\kappa(\psi,\phi)$.
\end{thebibliography}
\end{document}